\documentclass[a4paper]{ifacconf}

\usepackage{amsmath, amssymb, amsfonts}
\usepackage{bbm}
\allowdisplaybreaks
\usepackage{graphicx}      
\usepackage{natbib}        
\usepackage{subcaption}
\usepackage{tikz}
\usetikzlibrary{arrows.meta}
\usepackage{xcolor}
\definecolor{dgreen}{rgb}{0.,0.8,0.}

\newcommand{\qedsymbol}{$\blacksquare$}

\newtheorem{lemma}[thm]{Lemma}
\newtheorem{proposition}[thm]{Proposition}
\newtheorem{corollary}[thm]{Corollary}

\newtheorem{assumption}[thm]{Assumption}
\newtheorem{remark}[thm]{Remark}

\newenvironment{proof}[1][Proof]{%
  \par\noindent\textbf{#1. }\rmfamily
}{\hfill\qedsymbol\par}
\newcounter{example}
\renewcommand{\theexample}{\arabic{example}} 
\newenvironment{example}[1][Example]{%
  \refstepcounter{example}%
  \par\noindent\textbf{#1~\theexample. }\rmfamily
}{\hfill\qedsymbol\par}

\newcommand*{\dif}{\mathop{}\!\mathrm{d}}
\newcommand{\reals}{\mathbb{R}}
\newcommand{\naturals}{\mathbb{N}}
\newcommand{\D}{\mathbb{D}}
\newcommand{\diag}{\mathrm{diag}}

\newcommand{\sbs}[1]{_{\scriptscriptstyle \mathrm{#1}}}

\begin{document}
\begin{frontmatter}

\title{Optimal Control Synthesis of Closed-Loop Recommendation Systems over Social Networks%
\thanksref{footnoteinfo}} 

\thanks[footnoteinfo]{This work has been supported in part by ANR project Feeding Bias (ANR-22-CE380017-01).}

\author[First,Second]{Simone Mariano} 
\author[First]{Paolo Frasca} 

\address[First]{Univ.\ Grenoble Alpes, CNRS, Inria, Grenoble INP, GIPSA-lab, 38000 Grenoble, France (e-mail: simone.mariano@gipsa-lab.fr, paolo.frasca@gipsa-lab.fr).}
\address[Second]{Univ.\ Grenoble Alpes, CNRS, Sciences Po Grenoble-UGA, Pacte, 38000 Grenoble, France.}

\begin{abstract}
This paper addresses the problem of designing recommendation systems for social networks and e-commerce platforms from a control-theoretic perspective. We model the design of recommendation systems as a state-feedback infinite-horizon optimal control problem with a performance index that (i) rewards alignment/engagement, (ii) penalizes polarization and large deviations from an uncontrolled baseline, and (iii) regularizes exposure across neighboring users. 
The recommendation entries are fed to the platform users, who are assumed to follow a networked, multi-topic, continuous-time opinion dynamics. 
We show that the designed control yields a stabilizing recommendation system under simple algebraic/spectral conditions
on the weights that encode the platform’s preference for engagement, stability of preferences, polarization, and cross-user diversity. Conversely, we show that when ill-posed weights are selected in the optimal control problem (namely, when engagement is excessively rewarded), then the closed-loop system can exhibit destabilizing, pathological behaviors that conflict with the design objectives.
\end{abstract}

\begin{keyword}
Social networks and opinion dynamics; Control of networks; Applications of optimal control; Recommendation systems;
\end{keyword}

\end{frontmatter}
\section{Introduction}
In social networking and e-commerce platforms, personalized suggestions that optimize short-term engagement interact with the confirmation bias of the users  and can produce various ``information disorders'', such as extreme, polarizing opinions and echo chambers \citep{gausen2022abm,huszar2022algorithmic,castaldo2022junk}. 
Previous work has suggested that recommendation systems, when tuned purely for engagement, steer users toward extreme positions, and that extreme positions, in turn, increase engagement. In the single-user model by~\cite{FrascaRecommendation2022}, this mechanism is formalized in a closed loop, providing both a baseline mechanism and a benchmark phenomenon. Related theory confirms that, without sustained exploration, positions can drift to extremes~\citep{jiang2019degenerate}, and repeated learning on self-influenced logs can reduce diversity and long-run utility~\citep{chaney2018confounding}.

Despite evidence of unwanted effects of recommendations, mitigation is typically addressed only after deployment. The recent survey by~\cite{dean2024accounting} argues that most recommendation pipelines are designed without an explicit model of how users and algorithms shape one another and highlights several structural issues: memory-less architectures \citep{covington2016youtube}; simplified or omitted user dynamics and creator adaptation; and optimization aggressively oriented toward engagement \citep{chen2019topk,immorlica2024clickbait}. Mitigation is then performed {\it a posteriori} on logged data, which might obscure causal effects and favor myopic, symptom-level fixes \citep{sinha2016deconvolving,chaney2018confounding}, in turn degrading platform performance even more and inducing user- and society-level harms, including shifts in exposure composition, reduced content diversity, and polarization  \citep{mansoury2020feedback,nguyen2014exploring}.

In this paper, we aim to overcome these limitations by treating the design of recommendation systems as an infinite-horizon optimal control problem posed at the level of the user-platform interaction. We begin from a performance index that makes explicit the design trade-offs, since it (i) rewards alignment between recommendations and user opinions as a proxy for engagement, (ii) penalizes polarization and large deviations from an uncontrolled baseline, and (iii) regularizes exposure across neighboring users in the social graph. These terms define a family of weights that encode the platform’s normative choices about engagement, stability of preferences, polarization, and cross-user diversity. 


The recommendations then enter, as control input, a continuous-time opinion dynamics model that encompasses key aspects of the platform’s behavior. Our analysis adopts a multi-topic opinion-formation model, akin to~\cite{friedkin2016network}, in the continuous-time version proposed by \cite{ye2020continuous}. Social interactions are represented by a graph Laplacian, the inter-topic logic is captured by a suitable coupling matrix, and an anchoring term models the retention of the inner beliefs. 

Our analysis is model-based and yields explicit algebraic and spectral conditions under which the closed-loop recommendation–user system is stable and admits a unique equilibrium with a transparent dependence on the design weights. 
Conversely, when ill-posed weights are selected, e.g., when engagement is overemphasized relative to regularization and polarization penalties, we prove that the same closed-loop architecture can exhibit pathological behaviors that conflict with the design objectives, such as unbounded growth of opinions or amplification of disagreement, or even the absence of an optimal recommendation input. In this way, the control-theoretic viewpoint both clarifies when recommendation objectives are well-posed and provides tools to design interaction mechanisms with provable guarantees, upstream of deployment and in line with the interaction-focused agenda of~\cite{dean2024accounting}.

The remainder of the paper is organized as follows. In Section~\ref{sec:pi} we introduce the performance index and formulate the infinite-horizon optimal control problem used to generate the recommendation inputs. Section~\ref{sec:model} presents the model of the opinion dynamics of the users. In Section~\ref{sec:sols_gen}, we analyze the resulting LQ problem and characterize how the selection of the weights in the performance index affects the solution of the LQ problem. Section~\ref{sec:problems} highlights the pathological behaviors that can arise when the design weights in the performance index are chosen inappropriately. Finally, Section~\ref{sec:disc} discusses the modeling and design choices, summarizes the main findings and their limitations, and outlines possible extensions.

\noindent\textbf{Notation:} $\reals$ and $\naturals$ denote the sets of real and natural numbers, respectively.
For $n\in\naturals$, $\mathbf{1}_n$ is the $n$-dimensional vector of ones, and 
$[n]:=\{1,\dots,n\}$. $I_n$ and $O_n$, $n\in\naturals$, denote identity and zero matrices of appropriate dimensions. For $A,B \in \reals^{n \times m}$, $n,m \in \naturals$, $A \odot B$ denotes their Hadamard (elementwise) product, i.e., $(A \odot B)_{ij} := A_{ij} B_{ij}$ for all $i\in [n]$ and $j\in [m]$.
For $A,B\in\reals^{n \times n}$, $ n\in \naturals$, the relations $A\succeq B$ and $A\preceq B$ between symmetric matrices denote the
Loewner order, i.e., $A-B$ is positive semidefinite and negative semidefinite, respectively, while $A \succ 0$ and $A \prec 0$ ($A \succeq 0$  and $A \preceq 0$) denote (semi)positive and (semi)negative definite matrices, respectively.
For $n\in\naturals$, $\mathbb{D}^n$ is the set of real diagonal matrices with $n$ diagonal entries, while
$\mathbb{D}_{\succ0}^n$ (resp.\ $\mathbb{D}_{\succeq0}^n$) is the set of diagonal positive
definite (positive semidefinite) matrices with $n$ diagonal entries.
For a diagonal matrix $W_{(\cdot)}=\mathrm{diag}(w_{(\cdot),1},\dots,w_{(\cdot),nm})$, where $(\cdot)$ is a suitable label,
we write
$\lambda_{m,(\cdot)}:=\min_{i\in\{1,\dots,nm\}} w_{(\cdot),i}$ and
$\lambda_{M,(\cdot)}:=\max_{i\in\{1,\dots,nm\}} w_{(\cdot),i}$,
where the subscripts $m$ and $M$ denote minimum and maximum elements, respectively.
For a complex number $\lambda$, $\Re(\lambda)$ denotes its real part, and for a matrix $M$,
$\sigma(M)$ denotes its spectrum (set of eigenvalues).



\section{Performance index and design goal}
\label{sec:pi}
The first objective of this work is to clearly define a performance index that quantifies the distortion and polarization-versus-engagement issue identified in \cite{FrascaRecommendation2022} in a networked, multi-topic setting. The guiding principle in designing the recommendation system is simple: engagement should be rewarded, but only to the extent that it does not generate pathological dynamics. 

Consider a set of $n \in \naturals$ users holding opinions on $m \in \naturals$ topics, connected via a directed, weighted and connected graph $\mathcal{G}(\mathcal{E},\mathcal{V})$, where $\mathcal{E}$ and $\mathcal{V}$ are, respectively, the edge and vertex set, and with associated Laplacian matrix $L\in \reals^{n\times n}$, with $L=\Delta-\Gamma$, where $\Gamma\in\mathbb{R}^{n\times n}$ is the adjacency matrix and $\Delta= \diag(\Gamma \mathrm \mathbf{1}_n) $ is the degree matrix.
For any time $t\ge 0$, let $x(t)\in \reals^{nm}$ be the vector of opinions and $u(t)\in \reals^{nm}$ be the vector of the inputs provided by the recommendation systems. The entries $x_{(k-1)\cdot n + i}$ and $u_{(k-1)\cdot n + i}$, $k\in\{1, \dots,m\}$, $i\in\{1, \dots,n\}$, of $x$ and $u$ correspond to the opinion of the $i$-th user on the $k$-th topics and its corresponding input, respectively. 

The trade-offs that have been mentioned in the Introduction will be precisely quantified through the integral cost 
\begin{align}
J(x(t),u(t))\nonumber 
:= &\int_{0}^{\infty}\big(-J\sbs{EN}+J\sbs{P}+J\sbs{D}+J\sbs{EX}+J\sbs{F}\big)\,\dif t \nonumber \\
=&\int_{0}^{\infty} \ell(x(t),u(t))\,\dif t,
\label{eq:J}
\end{align} 
whose different terms correspond to the various objectives of the recommendation system.

The first term
\begin{equation*}
J\sbs{EN}:=x^{\top} W\sbs{EN} u,
\label{eq:J_EN}
\end{equation*}
$W\sbs{EN} \!  \in \! \mathbb{D}_{\succeq0}^{nm}$, models the user engagement by rewarding the alignment between user opinion and recommendation, defined as a weighted scalar product between vectors in $\reals^{nm}$, consistently with \cite{FrascaRecommendation2022,sprenger2024control,chandrasekaran2026network}.
Notice that by convention, the index is minimized, and hence engagement appears with a negative sign.
%

The term
\begin{equation*}
J\sbs{P}:=x^T W\sbs{P}x ,
\label{eq:J_P}
\end{equation*}
$W\sbs{P} \!  \in \! \mathbb{D}_{\succeq0}^{nm}$, penalizes the polarization of opinions, as done in network-aware designs and user-based evaluations~\citep{gausen2022abm,chandrasekaran2026network}. 

The term
\begin{equation*}
J\sbs{D}:=(x-x_{\mathrm{eq}})^T W\sbs{D}(x-x_{\mathrm{eq}}),
\label{eq:J_D}
\end{equation*}
$W\sbs{D} \!  \in \! \mathbb{D}_{\succeq0}^{nm}$ captures how much opinions deviate from the uncontrolled equilibrium  $x_{\mathrm{eq}}\in \reals^{nm}$, which generalizes the inner belief idea in~\cite{FrascaRecommendation2022} and \cite{friedkin2016network} and aims to preserve the alignment of the users' inner beliefs and its expressed opinion on a given topic. 

The term
\begin{equation*}
J\sbs{EX}:=u^T W\sbs{EX}u,
\label{eq:J_EX}
\end{equation*}
$W\sbs{EX} \!  \in \! \mathbb{D}_{\succ0}^{nm}$, captures the effort of the control and has the scope of limiting excessively strong or frequent inputs to avoid overexposure of the users in the social platform \citep{qin2024too}. 

Finally, the term 
\begin{equation*}
J\sbs{F}:=\alpha\sbs{F}\,u^\top L_u u,
\qquad
L_u:=I_m\otimes L_b, \qquad \alpha\sbs{F}\geq 0,
\end{equation*}
with $L_b:=\Delta_b-\Gamma_b$,  $\Gamma_b:=\tfrac{\Gamma+\Gamma^\top}{2}$ and $\Delta_b:=\operatorname{diag}(\Gamma_b\mathbf{1})$, is used to mimic collaborative filtering by regularizing exposure across neighboring users and robustifies the design without imposing hard constraints. 
This design choice is justified under the widely-supported assumption that users who interact with one another tend to share similar preferences and opinions \citep{mcpherson2001birds}. 

\begin{figure}[tbp]
  \centering
  \raisebox{0.75ex}{\scalebox{1.20}{\begin{tikzpicture}[
      >=Latex, line width=0.8pt, line join=miter, line cap=butt,
      every node/.style={font=\scriptsize},
      block/.style={draw, thick, align=center,
                    minimum width=2.2cm, minimum height=0.9cm, inner sep=2pt}
  ]
    \node[block] (user) at (0,1.4)   {User dynamics\\ $f(x(t),u(t))$};
    \node[block] (rec)  at (0,0)     {Recommendation system\\$\pi(x(t),u(t))$};

    \draw[->] (rec.north) -- node[right,pos=0.55] {$u$} (user.south);

    \path (user.east) ++(0, 0.20) coordinate (yout); 
    \path (user.east) ++(0,-0.20) coordinate (xfb);  

    \draw[->] (yout) -- ++(0.9,0) node[above,pos=0.15] {$ \qquad x \odot u$};

    \draw[->] (xfb) -- ++(0.9,0) node[below,pos=0.30] {$x$}
               |- (rec.east);
  \end{tikzpicture}}}
  \caption{Representation of user-recommendation system feedback loop.}
  \label{fig:sys}
\end{figure}

The design goal is to define a recommendation system that selects the appropriate $u$ so that the population follows a well-behaved trajectory and settles near a desirable steady regime, while minimizing~\eqref{eq:J}, see Figure~\ref{fig:sys}. This objective translates into the infinite-horizon optimal control problem
\begin{align}
   \min_{u}\quad &  J=\int_{0}^{\infty} \ell(x(t),u(t))\,\dif t, \nonumber
  \\
  \mathrm{s.t.}\quad & \dot x = f(x,u) ,\ \ x(0)=x_0.
\label{eq:prob_full}
\end{align}

In what follows, we (i) propose a tractable yet expressive model of the opinion dynamics, (ii) show how poor parameter choices and mis-specified objectives can induce undesirable opinion dynamics, and (iii) present conditions under which the closed-loop recommendation–user system is stable and admits a unique equilibrium.

\section{Opinion dynamics model}
\label{sec:model}
To define the users' dynamics, we follow the models presented in \cite{friedkin2016network,ye2020continuous}. Let $X \in \reals^{n\times m}$ be the opinion matrix, where $X_{ik}$ denotes the opinion of user $i\in\{1,\dots,n\}$ on topic $k\in\{1,\dots,m\}$. The continuous-time dynamics are
\begin{align}
    \dot{X}(t) = - & L\,X(t) \;-\; A_a\big(X(t)-X_\circ\big) \nonumber \\
    &+\; \big(U(t)-X(t)\big) \;+\; \big(X(t)C^\top - X(t)\big),
\label{eq:model_matrix_noninv}
\end{align}
with the single opinion evolving by 
\begin{align*}
    \dot x_{ik}
:= & -\sum_{j=1}^{n}\big( L_{ij}\,x_{jk}(t)\big)
  - A_{a,ii}\big(x_{ik}(t)-X^\circ_{ik}\big) \nonumber  \\ 
  &+ (u_{ik}(t) - x_{ik}(t)) 
  + \Big(\sum_{h=1}^{m} x_{i h}(t)\,C_{kh} - x_{ik}(t)\Big). 
\end{align*}
Matrix $L\in\reals^{n\times n}$ is the Laplacian of $\mathcal{G}(\mathcal{E},\mathcal{V})$, which drives consensus among neighboring users (within each topic), $C\in\reals^{m\times m}$ captures inter-topic influence within each user's opinions, while $A_a\in\D^{n\times n}_{\succ0}$ is a diagonal anchoring matrix with  $X_\circ\in\reals^{n\times m}$ collecting anchoring opinions, namely, the inner beliefs of the users on a given topic. Finally, the matrix $U(t)\in\reals^{n\times m}$ is the input provided by the recommendation system and {\color{black} appears in the relative form $(U-X)$, so that aligned recommendations induce no artificial amplification}.
%
 %
The inter-topic matrix $C$ satisfies the following property, consistent with Assumption~1 in \cite{ye2020continuous}, which prevents instability of the uncontrolled system.
 

\begin{assumption}
\label{ass:C}
Matrix $C$ is such that $c_{ii} \geq 0$, for all $i \in [m]$, $|c_{ij}| \leq 1 $, for all $i, j \in [m]$, and given $A:=C-I$, then $0$ is a semisimple eigenvalue of $A$ with multiplicity $p\geq 1$, while for any 
$\lambda\in\sigma(A)$ such that $\lambda \neq 0$ then $\Re(\lambda)<0$.
\end{assumption}

Equation~\eqref{eq:model_matrix_noninv} extends \cite{friedkin2016network,ye2020continuous} by including a recommendation input $U$ that specifies a stance per user and topic. At the same time, it extends the single-user closed-loop model of \cite{FrascaRecommendation2022} to a networked, multi-topic, continuous-time setting. 

We now provide a convenient vectorized form of \eqref{eq:model_matrix_noninv}.  
Let $x:=\mathrm{vec}(X)\in\reals^{nm}$ and $u:=\mathrm{vec}(U)\in\reals^{nm}$. Using the identity $\mathrm{vec}(AXB)=(B^\top\otimes A)\mathrm{vec}(X)$, one obtains
\begin{align}
\label{eq:aff_dyn}
\dot{x}
:= &\Big[(C\otimes I_n) - (I_m\otimes(L+A_a)) - 2I_{nm}\Big]x \nonumber \\
\;&+\; (I_m\otimes A_a)\,\mathrm{vec}(X_\circ) \;+\; u  =   A_c x + d + u  =:  f(x,u),
\end{align}
with
\begin{equation}
A_c \! := \! (C\otimes I_n) \!-\! \big(I_m\!\otimes\!(L + A_a + 2I_n)\big),
\, 
d \! := \! (I_m\otimes A_a)\,\mathrm{vec}(X_\circ).
\label{eq:model_vec_comp}
\end{equation}
The uncontrolled equilibrium solves $A_{uc}x_{\mathrm{eq}}+d=0$ with
\begin{equation*}
A_{uc} := (C\otimes I_n) - \big(I_m\otimes(L + A_a + I_n)\big),
\end{equation*}
and
\begin{equation*}
x_{\mathrm{eq}} := -A_{uc}^{-1}d,
\end{equation*}
and where $A_{uc}$ is Hurwitz under Assumption~\ref{ass:C}, see Lemma~2 in \cite{ye2020continuous}. In particular, when $X_\circ$ is constant, the uncontrolled dynamics $\dot x = A_{uc}x + d$ are stable and converge to $x_{\mathrm{eq}}$. 
The next result also builds on  Lemma~2 in \cite{ye2020continuous}. 
\begin{proposition}
\label{prop:Ac_Hur}
Given Assumption~\ref{ass:C}, $A_c$ in \eqref{eq:model_vec_comp} is Hurwitz.
\end{proposition}

\begin{proof}
Let $M:=L+A_a+2I_n$, with $L$ and $A_a$ as in \eqref{eq:model_matrix_noninv}, so that
$A_c=C\otimes I_n - I_m\otimes M$. Notice that, in view of Assumption~\ref{ass:C},
$\max_{\lambda\in\sigma(C)}\Re(\lambda)\leq 1$. For each $i\in[n]$, the $i$-th Gershgorin disc of $M$ is centered at
$M_{ii}=L_{ii}+(A_a)_{ii}+2$ with radius
$R_i=\sum_{j\ne i}|M_{ij}|=\sum_{j\ne i}|L_{ij}|\leq L_{ii}$.
Hence, every $\mu\in\sigma(M)$ satisfies
\begin{align*}
     \Re (\mu)  &\geq \min_{i\in[n]}\big\{(L_{ii}+(A_a)_{ii}+2)-R_i\big\}
 \\ &\geq \min_{i\in[n]}\{(A_a)_{ii}+2\}\ =:\ 2+a_{\min},   
\end{align*}
 where $a_{\min}:=\min_{i\in[n]} (A_a)_{ii}>0$. Take Schur decompositions $C=V T_C V^*$ and $M=Z T_M Z^*$ with $V\in \reals^{m \times m},Z\in \reals^{n \times n}$ unitary and $T_C\in \reals^{m \times m}$, $T_M\in \reals^{n \times n}$ upper triangular. Using the mixed‐product rule and unitary similarity,
$A_c = C\otimes I_n - I_m\otimes M = (V\otimes Z)\,\big(T_C\otimes I_n - I_m\otimes T_M\big)\,(V\otimes Z)^* $.
Thus $A_c$ is unitarily similar to an upper triangular matrix whose diagonal
entries are in $\mathcal{S}:=\{\lambda_i(C)-\mu_j(M): 1\leq i\leq m,\ 1\leq j\leq n\}$, and thus $\sigma(A_c)=\mathcal{S}$, counting algebraic multiplicities. Since $\big(\max_{\lambda\in\sigma(C)} \Re \lambda\big) - \big(\min_{\mu\in\sigma(M)} \Re \mu\big)
        \leq 1 - (2+a_{\min}) = -(1+a_{\min})$, $A_c$ is Hurwitz as claimed.
\end{proof}

We now turn to solving \eqref{eq:prob_full} with $f$ defined in \eqref{eq:aff_dyn}.

\section{Solutions of the LQ problem} 
\label{sec:sols_gen}
In this section, we define an optimization problem that is equivalent to \eqref{eq:prob_full}, and we study when its stage cost is positive (semi)definite under the proposed weight selection. We discuss the taxonomy of the tools applicable to find the solution of the designed optimal control problem from a LQ-theoretical perspective. 

The stage cost in \eqref{eq:J} can be rewritten as
\begin{equation}
\ell(x,u) \;=\; x^\top Q x \;+\; 2\,x^\top N u \;+\; u^\top R u \;+\; 2\,c^\top x,
\label{eq:stage_cost}
\end{equation}
with $Q=Q^{\top}=W\sbs{D}+W\sbs{P} \succeq 0$, $N=N^{\top}=-\tfrac12\,W\sbs{EN} \preceq 0$ and $R=R^{\top}=W\sbs{EX}+\alpha\sbs{F}L_u \succ 0$ being  positive and negative semidefinite (possibly diagonal) matrices of appropriate dimensions, and
\begin{equation*}
c=-W\sbs{D}x_{\mathrm{eq}}.
\end{equation*}
Notice that we have purposely ignored the constant component of the stage cost, which does not affect the optimizer. 

Defining $v:=u+R^{-1}N x$, the control problem \eqref{eq:prob_full} is equivalent to
\begin{align}
   \min_{v}\quad &  \widetilde J=\int_{0}^{\infty} \widetilde \ell(x(t),v(t))\,\dif t, \vspace{3pt}
  \nonumber\\
   \mathrm{s.t.}\quad & \dot x = \widetilde{f}(x,v) ,\ \ x(0)=x_0,  \label{eq:prob_full_alt}
\end{align}
with
\begin{equation}
\label{eq:subs}
    \widetilde A:=A_c - R^{-1}N,\qquad \widetilde Q:=Q-N R^{-1}N,
\end{equation}
and where $\widetilde{f}(x,v):=\widetilde A x + v + d$
and $\widetilde\ell(x,v)=x^\top \widetilde Q x + v^\top R v + 2\,c^\top x$, and with optimal control input $u^\star(x)=v^\star(x)-R^{-1}N x$.

We can then show that the quadratic form 
\begin{equation}
\widetilde\ell_{sq}(x,v) \;:=\; x^\top \widetilde Q x  \;+\; v^\top R v,
\label{eq:stage_cost_qp}
\end{equation}
is sign-definite under simple spectral bounds.
\begin{lemma}
\label{lem:PDP}
Consider $\widetilde\ell_{sq}$ in \eqref{eq:stage_cost_qp}. If
\begin{equation}
\lambda_{m,\mathrm{D}}+\lambda_{m,\mathrm{P}}>\frac{\lambda_{M,\mathrm{EN}}^2}{4 \lambda_{m,\mathrm{EX}}},
\label{eq:cond_L1}
\end{equation}
then the block matrix $\begin{bmatrix}\widetilde Q& O_{nm} \\ O_{nm} &R\end{bmatrix}$and thus $\widetilde\ell_{sq}(x,v)$ in \eqref{eq:stage_cost_qp} are positive definite.
\end{lemma}
\begin{proof} 
Given $R \succ 0$,
\begin{equation}
\label{eq:proof_L1_1}
\begin{bmatrix}\widetilde Q&O_{nm}\\O_{nm}&R\end{bmatrix}\succ 0 \iff \widetilde Q\succ 0 \land R \succ 0.
\end{equation}
$\widetilde Q\succ 0$ translates to, being $L_u \succeq 0$,
\begin{align*}
& \nonumber (W\sbs{D}+W\sbs{P}) - \frac{1}{4}W\sbs{EN}(\alpha\sbs{F}L_u+ W\sbs{EX})^{-1} W\sbs{EN} \\&\succ \Bigg(\lambda_{m,\mathrm{D}}+\lambda_{m,\mathrm{P}}-\frac{\lambda_{M,\mathrm{EN}}^2}{4\lambda_{m,\mathrm{EX}}}\Bigg) I_{nm}
\succ 0,
\end{align*}
and the last inequality is true if \eqref{eq:cond_L1} holds.
\end{proof}

\begin{corollary}
\label{cor:PDP2}
Suppose $Q,R$ and $N$
 are simultaneously orthogonally diagonalizable, that is,
there exist an orthogonal $U\in \reals^{nm \times nm}$ and $q,r\in \reals^{nm}_{>0}$, $w_{\mathrm{EN}}\in \reals^{nm}_{\geq 0} $ such that
$Q=U\mathrm{diag}(q)U^\top$, $R=U\mathrm{diag}(r)U^\top$,
$W_{\mathrm{EN}}=U\mathrm{diag}(w_{\mathrm{EN}})U^\top$.
Then
\begin{equation*}
\begin{bmatrix}\widetilde Q & O_{nm}\\O_{nm} & R\end{bmatrix}\succ0
\iff
q_i \;>\; \frac{w_{\mathrm{EN},i}^2}{4\,r_i}\ \ \forall i\in [nm].
\end{equation*}
\end{corollary}
\begin{proof}
Condition \eqref{eq:proof_L1_1} can be obtained by Schur complement. The right-hand side of \eqref{eq:proof_L1_1} can be rewritten equivalently, being simultaneously orthogonally diagonalizable, as $ U(\diag(q) - \frac{1}{4}(\diag(w_{\mathrm{EN}})(\mathrm{diag}(r))^{-1} (\diag(w_{\mathrm{EN}}) )U^\top \succ 0$, which translates to the
 condition $q_i \;>\; \frac{w_{\mathrm{EN},i}^2}{4\,r_i}$, for all $i\in [nm]$.
\end{proof}

\begin{remark}
If the weights are homogeneous for each topic, i.e., $Q = Q_t \otimes I_n, W_{\mathrm{EN}} = W_t \otimes I_n, R = R_t \otimes I_n + \alpha_{\mathrm{F}}\, I_m \otimes L_{\mathrm{sym}}$,
then $Q$, $W_{\mathrm{EN}}$, and $R$ pairwise commute since it holds that, given $A\in\reals^{m\times m}, \, B\in\reals^{n\times n}$, $(A\otimes I_n)(I_m\otimes B) = A\otimes B = (I_m\otimes B)(A\otimes I_n)$,
which in turn it implies that they are simultaneously orthogonally diagonalizable. Matrices $Q$, $R$, and $N$ are also simultaneously orthogonally diagonalizable when $J_F=0$
\end{remark}

\begin{remark}
\label{rem:semi}
Lemma~\ref{lem:PDP} and Corollary~\ref{cor:PDP2} can be restated with nonstrict inequalities in the bounds to obtain conditions for positive semidefiniteness of $\widetilde\ell_{sq}$. Indeed, $\diag(\widetilde Q,R)\succeq 0$ if, and only if $\widetilde Q\succeq 0$ and $R\succ 0$, with $R \succ 0$ being a standing assumption. Thus, replacing the strict bounds in Lemma~\ref{lem:PDP} and Corollary~\ref{cor:PDP2} by the corresponding nonstrict versions yields sufficient, and necessary and sufficient conditions, respectively, for $\widetilde\ell_{sq}$ to be nonnegative definite.
\end{remark}

Since $B=I_{nm}$, the pair $(\widetilde A,B)$ is controllable, and the infinite-horizon behavior is governed by the sign structure of the quadratic form $\widetilde\ell_{sq}$ in \eqref{eq:stage_cost_qp}, in line with classical LQR theory~\citep{anderson2007optimal,trentelman2002control}. We call the problem \emph{free-endpoint} when the limiting state is not prescribed, and \emph{fixed-endpoint} when convergence to a prescribed equilibrium, such as $x_{\mathrm{eq}}$ or the origin, is imposed. In the latter case, choosing a stabilizing feedback amounts to enforcing asymptotic convergence to the prescribed equilibrium.

If \(\widetilde\ell_{sq}\) is strictly positive definite, standard infinite-horizon LQ theory gives a unique stabilizing Riccati solution and a stabilizing optimal feedback. Hence a fixed-endpoint condition is redundant: convergence is already enforced by the objective. If \(\widetilde\ell_{sq}\) is nonnegative, stabilization additionally requires the unpenalized modes to be detectable, e.g., through \((\widetilde Q^{1/2},\widetilde A)\). Since \(A_c\) is Hurwitz by Proposition~\ref{prop:Ac_Hur}, sufficiently small engagement weights \(W_{\mathrm{EN}}\), or sufficiently large exposure weights in \(R\), keep \(\widetilde A=A_c-R^{-1}N\) Hurwitz by perturbation arguments~\cite[Ch.~VI]{horn2012matrix}. Thus the semidefinite case remains benign under mild engagement rewards whenever the nonstrict versions of Lemma~\ref{lem:PDP} and Corollary \ref{cor:PDP2} give \(\widetilde Q\succeq0\).

When \(\widetilde Q\) is indefinite, the situation changes. The unconstrained infinite-horizon problem may have no optimal input, both in fixed- and free-endpoint formulations~\cite[Thms.~4.1--5.1]{trentelman1989regular}. Even when a minimizer exists, it need not be stabilizing. In that regime, negative or unpenalized directions may be exploited to lower the cost, so convergence must be imposed externally rather than following from minimization of \(\widetilde J\) in \eqref{eq:prob_full_alt}.
For recommendation design, this means that the objective no longer aligns engagement, polarization avoidance, and stability. We therefore use Lemma~\ref{lem:PDP} and Corollary \ref{cor:PDP2} as design constraints on \((W_{\mathrm{D}},W_{\mathrm{P}},W_{\mathrm{EN}},W_{\mathrm{EX}},\alpha_{\mathrm{F}})\), restricting engagement rewards to values compatible with deviation, polarization, and exposure regularization. The following section illustrates the resulting failure modes.
\section{Free-endpoint LQ solutions: semidefinite and sign-indefinite cases}
\label{sec:problems}
When $\widetilde\ell_{sq}$ is not positive definite, we use the free-endpoint theory of~\citet{trentelman1989regular}. For simplicity, the examples below take $d=c=0\cdot\mathbf{1}_{nm}$. The relevant algebraic Riccati equation is
\begin{equation}
\label{eq:ARE-min}
\widetilde A^\top P+P\widetilde A-PR^{-1}P+\widetilde Q=0.
\end{equation}
In the semidefinite case,~\cite[Thm~4.2]{trentelman1989regular} yields the unique free-endpoint optimizer $v^\star=-R^{-1}P_\circ x$, where $P_\circ$ is the smallest positive-semidefinite solution of \eqref{eq:ARE-min}.

In the sign-indefinite case, let $P_-$ and $P_+$ denote the minimal antistabilizing and maximal stabilizing symmetric solutions of \eqref{eq:ARE-min}, respectively, and define
\[
A_-:=\widetilde A-R^{-1}P_-,\qquad
\mathcal N:=(\ker P_-\mid A_- )\cap\mathcal X^+(A_-),
\]
where $(\ker P_-\mid A_-)$ is the largest $A_-$-invariant subspace contained in $\ker P_-$ and $\mathcal X^+(A_-)$ is the $A_-$-invariant subspace associated with eigenvalues having nonnegative real part. If $\Pi_{\mathcal N}$ is the corresponding $A_-$-invariant projector, the supported curvature is
\begin{equation}
\label{eq:P-sup}
P_{\mathcal N}:=P_-\Pi_{\mathcal N}+P_+(I_{nm}-\Pi_{\mathcal N}).
\end{equation}
A finite value exists if an antistabilizing $P_-$ exists. Moreover, an optimal input exists for every initial state if and only if $\ker(P_+-P_-)\subset\ker P_-$; in that case
\[
v^\star(x)=-R^{-1}P_{\mathcal N}x.
\]
In the following examples, we show how having ill-posed weights in the definition of the performance index \eqref{eq:J} might lead to undesirable behaviors, such as unbounded growth of opinions or amplification of disagreement, or even the absence of an optimal recommendation input.

\begin{example}
Consider \eqref{eq:prob_full_alt} with $n=1$, $m=2$, $A_a=1$, $L=0$, $d=c=0$, $R=I_2$, $N=\diag(-2-\eta,-5/2+\eta)$, $Q=\diag\big((2+\eta)^2,(-5/2+\eta)^2-\beta\big)$, and 
\[
C=\begin{bmatrix}1&\xi\\0&1/2\end{bmatrix},
\]
where $0<\eta<5/2$, $\xi\neq0$, and $0<\beta<\eta^2$. Then
\[
\widetilde A=\begin{bmatrix}\eta&\xi\\0&-\eta\end{bmatrix},\qquad
\widetilde Q=\diag(0,-\beta),
\]
so the stage cost is sign-indefinite. Let $\Delta:=\sqrt{\eta^2-\beta}$. Expanding \eqref{eq:ARE-min} for $P=P^\top=[P_{ij}]$ gives
\begin{equation}
\label{ex1:ARE-system-short}
\begin{cases}
2\eta P_{11}-(P_{11}^2+P_{12}^2)=0,\\
\xi P_{11}-P_{12}(P_{11}+P_{22})=0,\\
2\xi P_{12}-2\eta P_{22}-(P_{12}^2+P_{22}^2)-\beta=0.
\end{cases}
\end{equation}
The diagonal branch of \eqref{ex1:ARE-system-short} gives $P=\diag(0,-\eta\pm\Delta)$. The minimal antistabilizing solution is therefore  $P_- =\diag(0,-\eta-\Delta)$,
because $A_-:=\widetilde A-P_-$ has eigenvalues $\eta$ and $\Delta$, both positive. The remaining, off-diagonal branch of \eqref{ex1:ARE-system-short} contains the stabilizing extremal solution $P_+$. Since $\ker P_-=\mathrm{span}(e_1)$ and $A_-e_1=\eta e_1$, we get $\mathcal N=\mathrm{span}(e_1)$. Hence $P_{\mathcal N}$ in \eqref{eq:P-sup} gives a unique free-endpoint feedback, but the closed loop acts as $A_-$ on $\mathcal N$ and therefore retains the unstable eigenvalue $\eta>0$. Thus the free-endpoint optimizer exists, but it is not stabilizing.
\end{example}

\begin{example}
Consider \eqref{eq:prob_full_alt} with $n=1$, $m=2$, $A_a=1$, $L=0$, $d=c=0$, $R=I_2$, $C=I_2$, $N=-3I_2$ and $Q=\diag(8,10)$.
Then $\widetilde A=I_2$ and $\widetilde Q=\diag(-1,1)$. For $P=P^\top=[P_{ij}]$, equation~\eqref{eq:ARE-min} becomes
\begin{equation}
\label{ex2:ARE-system-short}
\begin{cases}
-P_{11}^2+2P_{11}-P_{12}^2-1=0,\\
P_{12}(2-P_{11}-P_{22})=0,\\
-P_{22}^2+2P_{22}-P_{12}^2+1=0.
\end{cases}
\end{equation}
The case $P_{12}\neq0$ gives no real solution after eliminating $P_{22}$, so $P_{12}=0$. Thus $P_{11}=1$ and $P_{22}\in\{1-\sqrt2,1+\sqrt2\}$, namely $P_- =\diag(1,1-\sqrt2)$, and
$P_+ =\diag(1,1+\sqrt2)$. However, $\ker(P_+-P_-)=\mathrm{span}(e_1)\not\subset\ker(P_-)=\{0\}$.
Therefore the attainability condition of~\cite[Thm.~5.1]{trentelman1989regular} fails. For initial states with a nonzero component along $e_1$, the infimum is finite but no admissible input attains it.
\end{example}

\begin{example}
Consider \eqref{eq:prob_full_alt} with $n=1$, $m=2$, $A_a=1$, $L=0$, $d=c=0$, $R=I_2$, and
\[
C=\begin{bmatrix}1&\xi\\ \xi&1\end{bmatrix},\qquad
N=-\eta I_2,\qquad Q=\eta^2 I_2,
\]
where $\eta>2+|\xi|$ and $\xi\neq0$. Then
\[
\widetilde A=\begin{bmatrix}-2+\eta&\xi\\ \xi&-2+\eta\end{bmatrix},\qquad
\widetilde Q=O_2.
\]
Thus $\widetilde\ell_{sq}$ is positive semidefinite. Let
\[
T=\frac{1}{\sqrt2}\begin{bmatrix}1&1\\1&-1\end{bmatrix},\qquad
\widetilde A=T\diag(\lambda_1,\lambda_2)T^\top,
\]
with $\lambda_{1,2}=-2+\eta\pm |\xi|>0$. In these coordinates, \eqref{eq:ARE-min} reduces to two scalar equations, and the minimal positive-semidefinite solution is $P_\circ=O_2$. Hence the free-endpoint optimizer is $v^\star=0$. Since both eigenvalues of $\widetilde A$ are positive, the pair $(\widetilde Q^{1/2},\widetilde A)$ is not detectable and the uncontrolled unstable modes are invisible to the cost. Consequently, the optimal input does not counteract the opinion dynamics, and the state diverges.
\end{example}

The three examples illustrate complementary pathologies that arise outside the cone characterized by Lemma~\ref{lem:PDP} and Corollary~\ref{cor:PDP2}. In Example~1, the stage cost is sign-indefinite and the free-endpoint optimizer exists, but the resulting feedback leaves an unstable mode unregulated. This behavior is visible in the top image of  Fig.~\ref{fig:ex1-ex3}, where the corresponding closed-loop trajectory diverges. In Example~2, the optimal value of the cost is finite on a nontrivial set of initial conditions, but the infimum is not attained. In Example~3, the cost is positive semidefinite rather than indefinite, but detectability fails. In particular, the optimizer is the zero input, and the unstable state diverges, as also shown in the bottom image of  Fig.~\ref{fig:ex1-ex3}. Taken together, these cases show that rewarding engagement beyond the regularization supplied by polarization, deviation, mismatch, and exposure penalties can make the LQ synthesis either nonattainable or formally optimal but behaviorally incompatible with the design objective.

\begin{figure}[t]
    \centering
    \includegraphics[width=0.84\linewidth, trim=0 5pt 0 8pt , clip]{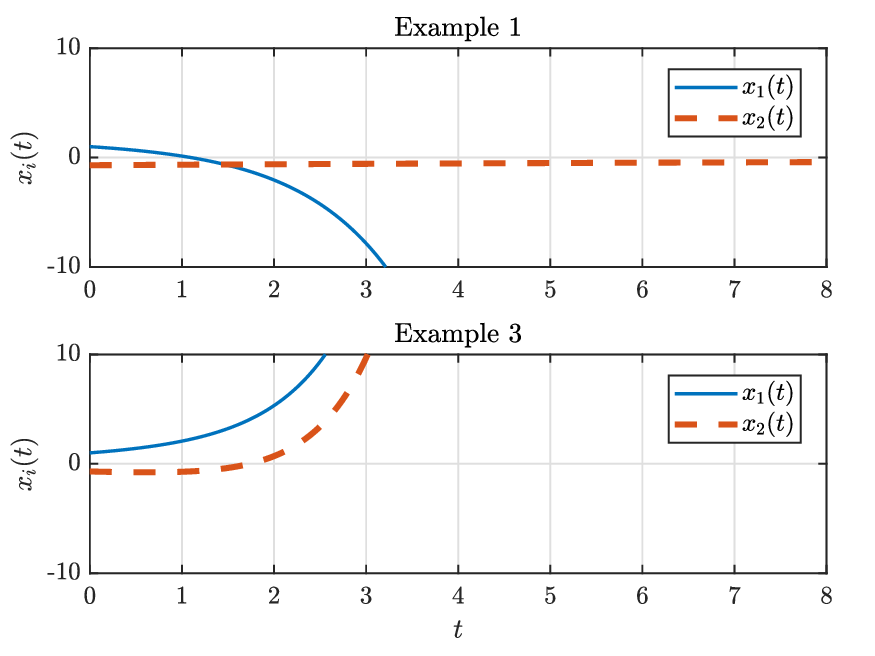}
     \vskip -12pt
    \caption{\vspace{-1pt} Closed-loop trajectories in Examples~1 and~3.}
    \label{fig:ex1-ex3}
\end{figure}

\section{Discussion and conclusion}
\label{sec:disc}

The main contribution of this work is to combine explicit opinion-dynamics modeling with closed-loop optimal control in order to derive spectral conditions for the well-posedness of recommendation design.
Our analysis shows that Lemma~\ref{lem:PDP} and Corollary \ref{cor:PDP2} should be interpreted as design constraints on the weights of the performance index. These conditions restrict the engagement reward relative to the penalties on polarization, deviation from baseline, and control effort. If engagement is over-weighted, the effective quadratic form becomes semidefinite or sign-indefinite, and the design enters a free-endpoint LQ regime in which stability is no longer guaranteed by the cost itself. Stability must then be imposed externally, which is undesirable: the objective no longer encodes the intended interaction between users and the platform.

The present work has several limitations that we would like to discuss.  
Some limitations originate from  our assumptions in modelling the social network and the interactions between the users and the recommendation system. Indeed, we assumed that the recommendation system can access full-state feedback and the social network has a known, time-invariant interaction structure, encoded by $L$, $C$, and $A_a$. In practice, these quantities must be estimated from sparse and noisy data and may be time-varying or uncertain. Extensions to partial observation, output feedback, uncertainty, and recursive model learning are therefore natural next steps.

Other limitations concern the class of recommendation designs.
Indeed, we considered quadratic penalties, static control laws, and unconstrained inputs. This controller class is analytically tractable and allows us to derive explicit Riccati-based conditions, but may be too narrow for the performance needs of recommendation systems. 
Furthermore, constraints on exposure, fairness, safety, item availability, and input variety cannot be captured by unconstrained static feedback laws. Enforcing such requirements may call for constrained or nonlinear formulations, dynamic controllers with memory, and time or state-dependent parameters in the cost objectives.

\bibliography{Bib}

\end{document}